\documentclass[12pt]{article}
\usepackage{amssymb}
\usepackage{amsmath, amsthm, geometry, csquotes}

\newtheorem{theorem}{Theorem}

\newtheorem*{axiom}{Axiom}

\newtheorem{claim}{Claim}

\newtheorem{lemma}{Lemma}

\newtheorem{remark}{Remark}

\input{tcilatex}
\begin{document}

\title{\textbf{Decision-making under risk: when is utility maximization
equivalent to risk minimization?}}
\author{Francesco Ruscitti\thanks{%
Department of Economics, John Cabot University, Via della Lungara 233, 00165
Rome, Italy; Email: fruscitti@johncabot.edu.} \and Ram Sewak Dubey\thanks{%
Department of Economics, Feliciano School of Business, Montclair State
University, Montclair, NJ 07043, USA; E-mail: dubeyr@montclair.edu} \and %
Giorgio Laguzzi\thanks{%
University of Eastern Piedmont, Department of Science and Technological
Innovation, Alessandria, Italy; Email: giorgio.laguzzi@uniupo.it.}}
\date{}
\maketitle

\begin{abstract}
Motivated by the analysis of a general optimal portfolio selection problem, which encompasses as special cases an optimal consumption and an optimal debt-arrangement problem, we are concerned with the questions of how a personality trait like risk-perception can be formalized and whether the two objectives of utility-maximization and risk-minimization can be both achieved simultaneously. 
We address these questions by developing an axiomatic foundation of preferences for which utility-maximization is equivalent to minimizing a utility-based shortfall risk measure. 
Our axiomatization hinges on a novel axiom in decision theory, namely the risk-perception axiom.

\smallskip

\noindent \emph{Journal of Economic Literature} Classification Numbers: 
\textit{C65, D81, D01.}
\end{abstract}

\section{Introduction}

Consider a two-period model of frictionless asset markets and assume that financial markets are incomplete. 
Suppose that a single decision-maker is endowed with an initial wealth $w_{0}$ and a random non-negative endowment vector $\omega $, where random means contingent on the state of nature that will unfold `tomorrow'. 
Assume that our decision-maker's objective is to choose an `optimal' financial position (available `tomorrow') from the set of feasible financial positions. 
A financial position is feasible if it is non-negative and there exists a least-cost super-replicating portfolio whose cost `today' does not exceed $w_{0}$. By super-replicating we mean a portfolio with a payoff `tomorrow'
which added to the endowment covers the targeted payoff profile (i.e., financial position).

The optimization problem outlined above (which will be formalized in the next section) matters, and we care about it, beacuse it can be interpreted in two different ways, as follows. 
First scenario: one can think of the decision-maker as a consumer who only cares about future consumption and wants to secure an `optimal' and budget-feasible consumption vector by trading assets `today' on incomplete financial markets. 
Note that consuming $\omega $ is always a feasible option (autarky). 
Second, alternative scenario: one can think of our decision-maker as an investor, or a banker, or a financial institution or portfolio-manager. 
The investor faces the following choice: to stay in autarky or to borrow an amount of money $w_{0}$ and use the loan to invest in some portfolios. 
If the investor stays put, `tomorrow' it will end up with its own random endowment $\omega $. 
If, instead, the investor borrows money, `tomorrow' it will have to use its own endowment together with the random payoff of an appropriate portfolio to repay the loan and, possibly, earn some profit. 
Suppose the investor prefers borrowing to staying in autarky. 
Then, the investor's problem is to decide `optimally' how much to pay back and under what future contingencies (a so-called financial position), i.e., the kind of debt-arrangement or contract to enter into. 
The debt-arrangement must be feasible, in the sense that there must exist a portfolio whose present cost does not exceed the amount borrowed and whose future payoff is enough to cover the liabilities taking the endowment into account. 
If it turns out that the investor's `optimal' future financial position is risk-free, one may interpret this as the investor issuing a bond `today'. 
If, instead, the investor's `optimal' future financial position is random, i.e., state-contingent, such arrangement may be interpreted as the decision maker selling a stock
`today'.

A question immediately arises as to what optimality criterion should be used by the investor and why. It should be obvious that in the first scenario the investor seeks to maximize its period-two utility function. 
However, arguably our investor may also be interested in lowering its exposure to risk by trading assets. One says that an agent lowers its exposure to risk if the optimal financial position is less risky than autarky, where the risk is quantified by an appropriate risk measure. 
This, in turn, begs the logical question of what risk measure should be used and whether the two objectives of utility-maximization and exposure-to-risk minimization can be both achieved simultaneously. 
In the second scenario, on the other hand, it may be rather controversial to assume that the investor's objective function is the canonical utility. 
This is because, by construction of the optimization problem, the control variable (i.e., a financial position interpreted as how much to pay back and in what state of nature) need not coincide with the investor's consumption. 
Could a way out of this problem be to assume that the objective function is a measure of risk used by the investor to determine the least risky way of borrowing money, investing, and paying off debt? 
This brings up the question of whether there exist meaningful conditions on preferences under which maximizing them is the same as minimizing risk, where risk is quantified according to some known risk measure. 
Henceforth, a known risk measure is one that either has attracted considerable research attention%
\footnote{E.g., the utilty-based shortfall risk measure. 
See, for example, Gundel and Weber (2007) or Guo and Xu (2019).} 
or is utilized by regulatory agencies and financial institutions as a tool for risk-management and bank capital requirements regulation.%
\footnote{E.g., the Value at Risk (VaR) or conditional value at risk (CVaR).}

The aim of this paper is to address the aforementioned questions.
We do so by axiomatizing preferences (defined on a payoff space, or space of financial positions) which can be represented by a real-valued utility function, say $U$, which is such that $-U$ is a known coherent risk measure.
Going back to the optimization problem sketched above, clearly in this case maximizing utility is equivalent to minimizing risk. 
Furthermore, if a solution exists, the decision-maker automatically lowers its exposure to risk. 
The gist of the proof of our main result (Theorem 1) can be explained as follows. 
We set out six axioms. 
The first five axioms are closely related to the ones introduced by Gilboa and Schmeidler (1989) and result in $-U$ being a coherent risk measure, where $U$ of any financial position can be written as the minimum expected value over a (non-empty and unique) closed and convex set of finitely additive probability measures on the states of nature. 
Call such a set $C$. 
Since the first five axioms do not tell us exactly what $C$ is, there are conceivably-many such sets. 
Hence, in the absence of further axioms there is indeterminacy about the exact functional form of the utility function. 
The sixth axiom, however, formalizes and reflects the decision-maker's risk-perception and resolves the indeterminacy. 
Specifically, the sixth axiom specifies the acceptance set associated with $-U$, %
\footnote{It is well-known that the acceptance set of a coherent risk measure is a convex cone containing the positive cone of the payoff space (see section $3$).} 
call it $P$ for future reference. 
This is equivalent to defining the upper contour set of the origin of the payoff vector space. 
Given homogeneity of degree one of $U$, this, in turn, determines indifference curves, hence the decision maker's preferences and, ultimately, the functional form of $U$. 
Interestingly, we show that choosing the acceptance set appropriately via the sixth axiom yields a $-U$ which is a known risk measure (see Remark 2 below). 
In fact, as will be clear in section $4$, the sixth axiom states that a financial position is acceptable if and only if its expected value is non-negative for every (subjective) probability measure in the decision-maker's closed and convex ambiguity set, say $D$.%
\footnote{The idea, here, is that the true probability measure is unknown. 
So, the decision-maker hedges against the risk arising from ambiguity regarding the true probability measure on the states of nature.} 
Now, some straightforward algebra shows that $U$ is the minimum expected value over $D$. 
By uniqueness of the set of probability measures, we obtain $C=D$. 
Therefore, the sixth axiom pins down the specific utilty function representing given preferences.
Furthermore, it turns out that $-U$ corresponds to a robust version of the utility-based shortfall risk measure (SR, henceforth); namely, the SR induced by the (decision maker's) ambiguity set, loss function, and risk threshold given, respectively, by $D$, the identity function, and zero.%
\footnote{See, e.g., Guo and Xu (2019) and Gundel and Weber (2007) for the definition and a fairly detailed account of utility-based shortfall risk measures. 
We note, in passing, that VaR is a special case of SR.}

Our notion of risk-perception (described in the sixth axiom) contributes to the literature on study of optimal response to alternate notions of risk. 
Kimball (1990) has developed a theory of the optimal response of decision variables to risk, and, in particular, a theory of precautionary saving within an expected-utility framework. 
More specifically, he defines \enquote{prudence} as a measure of the sensitivity of the optimal choice of a decision variable to risk to estimate precautionary saving as an example of the effect of risk on a decision variable. 
One could infer that higher precautionary saving in response to increase in risk is akin to higher demand for \enquote{insurance}. 
In other words, increasing demand for \enquote{insurance} as risk increases signifies increasing prudence level of the decision maker. 
This helps us intuitively relate our risk perception axiom to the notion of prudence developed in Kimball (1990). 
In Claim 2, we make this idea precise and show that in our model (based on axioms A.1-A.6) any individual with maximum risk-perception is prudent, i.e., is \enquote{fully insured}, or demands a risk-free (constant) financial position.

The remainder of the paper is organized as follows: in section $2$ we describe the underlying asset market model we have in mind and formalize the optimization problem (which drives our interest in investigating utility-maximization versus risk-minimization) by exploiting some known results concerning the fundamental theorem of asset pricing. 
Also, in Remark 1 we discuss the relation of our optimization problem to Gundel and Weber (2007). 
In section $3$ we gather some useful results regarding coherent risk measures in Banach spaces. 
In section $4$ we develop an axiomatization of preferences which can be represented by the negative of a utility-based shortfall risk measures and we state our main result (Theorem 1). 
The proof of the latter is relegated to the appendix. 
We also discuss our novel risk-perception axiom and, in Remark 3, we show how it can be connected to the notion of prudence introduced in Kimball (1990). 
In section $5$ we make some concluding remarks.

\section{Model setup and the motivating optimization problem\protect\medskip}

\noindent Asset markets and the optimization problem described informally in
the introduction can be modeled as follows.\smallskip\ 

\noindent Let $X$ be the space of payoffs, or financial positions, which is
a vector space ordered by a cone $X_{+}$. We will assume throughout the
paper that $X$ is $\ell _{\infty }$, $X_{+}$ is $\ell _{\infty }^{+}$ (the
pointwise ordering), and that $\ell _{\infty }$ is equipped with the
supremum norm. $\mathbf{1}$ denotes the constant sequence $\left( 1,1,\cdots
,1,\cdots \right) $ which we assume is the (risk-free) payoff of a tradeable
bond, like a treasury bill. We remind the reader that $\ell _{\infty }$
endowed with the supremum norm is a Banach space, $\mathbf{1}\in \ell
_{\infty }^{+}$ is an order unit of $\ell _{\infty }$, and $\ell _{\infty
}^{+}$ is closed. Incidentally, recall that\smallskip

\begin{equation*}
\ell _{\infty }=\left\{ \mathbf{x}=(x_{1},x_{2},\cdots ,x_{n},\cdots )\in 
\mathbb{R}^{\mathbb{N}}:\underset{n\in \mathbb{N}}{\sup }|x_{n}|\leq M_{%
\mathbf{x}}\right\} \text{,}
\end{equation*}%
\smallskip

\noindent where $M_{\mathbf{x}}$ is a positive number depending on $\mathbf{x%
}$. It is of some interest to note that we may as well let $X=C\left( \left[
0,1\right] \right) $, and assume that $C\left( \left[ 0,1\right] \right) $
is ordered pointwise and is equipped with the sup-norm. Accordingly, we
could let $\mathbf{1}$ be the constant function one on $\left[ 0,1\right] $
(i.e., $\mathbf{1}\left( t\right) =1$ for all $t\in \left[ 0,1\right] $),
which is an order unit of $C\left( \left[ 0,1\right] \right) $. $H=\mathbb{R}%
^{J}$ is the portfolio space, where $J$ is the finite number of assets.
Thus, a vector $h\in \mathbb{R}^{J}$ represents holdings of the available
assets. The operator $R:\mathbb{R}^{J}\rightarrow X$ is the payoff operator,
which is a linear operator: a portfolio $h\in \mathbb{R}^{J}$ traded `today'
yields the payoff $R\left( h\right) $ in $X$ `tomorrow'. Clearly, $R\left( 
\mathbb{R}^{J}\right) $ is a linear subspace of $X$ and is denoted by $Z.$
Asset markets are assumed to be incomplete, i.e., $Z\subset X$. $q:\mathbb{R}%
^{J}\rightarrow \mathbb{R}$ is the (portfolio) pricing functional, which is
a linear functional: given any $h\in \mathbb{R}^{J}$, $q\left( h\right) $ is
the cost of that portfolio. Clearly, the (portfolio) pricing functional is
fully characterized by the vector of exogenously-given asset prices. We
assume that $q:\mathbb{R}^{J}\rightarrow \mathbb{R}$ is arbitrage-free.%
\footnote{%
See, e.g., LeRoy and Werner (2014) for the definition of arbitrage-free
asset prices in simple two-period models of financial markets.} \smallskip

\noindent The decision-maker has an initial wealth $w_{0}>0$ (at $t=0$) and
an endowment vector $\omega \in X_{+}\backslash \left\{ \mathbf{0}\right\} $
(at $t=1$). We let the function $U:X\rightarrow \mathbb{R}$ be a numerical
representation of the decision-maker's preferences. We are ready to set up
the optimization problem discussed in the introduction. It is inspired by
the portfolio selection problem examined in Gundel and Weber (2007) but the
problem we present next is adapted to the environment we are focusing
on:
\begin{align}
\max_{\mathbf{c}\in X_{+}}\; & U\left(\mathbf{c}\right)\notag \\ 
\text{subject to:}\; &\inf \left\{q\left( h\right): h\in \mathbb{R}^{J}\;\text{and}\; \mathbf{c}\leq \omega + R\left(h\right) \right\} \leq w_{0}. \tag{1}
\end{align}
Note that $\omega $ is feasible for program $(1)$, thus the
feasible set is non-empty. Since $\mathbf{1}\in X_{+}$ is an order unit of $%
X $ and, by assumption, it is also the payoff of a portfolio including only
the treasury bill, there exists a $\bar{h}\in \mathbb{R}^{J}$\ such that $%
R\left( \bar{h}\right) =\mathbf{1}$\ is an interior point of $X_{+}$. In
turn, this implies, by the fundamental theorem of asset pricing, that the
set of valuation functionals $VF$ is non-empty.\footnote{%
See, e.g., Ross (1978) or Clark (1993) for the valuation problem and the
fundamental theorem of asset pricing in ordered vector spaces.} Recall that
a valuation functional $f:X\rightarrow \mathbb{R}$, roughly, is a positive,
continuous, and linear functional\ such that $f\left( \mathbf{x}\right)
=q\left( h\right) $ for all $\mathbf{x}\in Z$, where $h\in \mathbb{R}^{J}$
and $\mathbf{x}=R\left( h\right) $. It's easy to see that since $\mathbf{1}$
is an order unit and $R$ is linear, for every $\mathbf{c}\in X$ there exists
a portfolio $h\in \mathbb{R}^{J}$ which super-replicates $\mathbf{c}-\omega $%
, i.e., $\mathbf{c}\leq \omega +R\left( h\right) $. Moreover, the cost $%
q\left( h\right) $ of any such portfolio is bounded below by $f\left( 
\mathbf{c}-\omega \right) $, where $f$ is any element of $VF$. 
Therefore, it is not difficult to see that the $\inf $ in program $(1)$ defines a real-valued function on $X_{+}$ (known as the super-replication price) and
\begin{equation}
\inf \left\{q\left( h\right) :h\in \mathbb{R}^{J}\text{ and }\mathbf{c}\leq \omega +R\left( h\right) \right\} \geq \sup \left\{f\left( \mathbf{c} - \omega \right): f\in VF\right\}  \tag{2}
\end{equation}%
Finally, we put on record (without proof) a result which will come
in handy for the proof of Claim 2 below: it is possible to show that the
feasible set of program $(1)$ is convex and such that $\mathbf{c}_{1}\leq 
\mathbf{c}_{2}$ and $\mathbf{c}_{2}$ feasible imply that $\mathbf{c}_{1}$ is
feasible as well.\smallskip

\begin{remark}
\emph{In a simple two-period model with no arbitrage opportunities and in which the optimal solution to program $(1)$ can be exactly replicated, the optimization problem studied by Gundel and Weber (2007) reduces to the above program $(1)$. 
To see this, note that by virtue of exact replication $(2)$ above can be assumed to hold with equality, which is then the Gundel and Weber's budget constraint. 
Moreover, given our axiomatization in section $4$ and the proof of Theorem 1 in the appendix, it turns out that the objective function in program $(1)$ is essentially the same as Gundel and Weber's, and maximizing it is the same as minimizing a SR risk measure. 
Consequently, since the zero vector is feasible for program $(1)$, the risk constraint in Gundel and Weber's problem becomes redundant, thus it need not be accounted for.}
\end{remark}

\section{Background: risk measures and some useful results}

\noindent In this section we gather some facts and concepts, about coherent
risk measures, that will be used in section $4$ and in the appendix. They
are taken and rearranged from Kountzakis et al. (2013) to which we refer the
reader for a detailed account.\smallskip

\noindent Let $X$ be a Banach space ordered by the (non-trivial) cone $P$,
and let $x_{0}\in P$ be an order unit of $X$. A function $\rho :X\rightarrow 
\mathbb{R}$ is said to be a coherent risk measure if it satisfies the
following properties:\smallskip

\begin{enumerate}
\item {$x\geq y$ implies $\rho \left( x\right) \leq \rho \left( y\right) $.}

\item {$\rho \left( x+tx_{0}\right) =\rho \left( x\right) -t$ for any $t\in 
\mathbb{R}$.}

\item {$\rho \left( x+y\right) \leq \rho \left( x\right) +\rho \left(
y\right) $ for every $x,y\in X$.}

\item {$\rho \left( \lambda x\right) =\lambda \rho \left( x\right) $ for
every $x\in $}$X${\ and every real number $\lambda \geq 0$.\smallskip }
\end{enumerate}

\noindent Suppose we are given a coherent risk measure $\rho $. Let $%
\mathcal{A}_{\rho }=\left\{ x\in X:\rho \left( x\right) \leq 0\right\} $,
which is the so-called acceptance set associated with $\rho $. It turns out
that $\mathcal{A}_{\rho }$ is a cone of $X$ which contains $P$, and for each 
$x\in X$ we have\smallskip

\begin{equation}
\rho \left( x\right) =\inf \left\{ t\in \mathbb{R}:x+tx_{0}\in \mathcal{A}_{\rho }\right\} .  \tag{3}
\end{equation}%
\smallskip

\noindent In the proof of our Theorem 1 we will exploit the following useful
result (see Theorem 1 in Kountzakis et al., 2013):\smallskip

\begin{lemma}
\textit{Suppose that} $X$\ \textit{is a Banach space ordered by the cone} $P$. 
\textit{If} $x_{0}$\ \textit{is an order unit of }$X$, \textit{then the function}\smallskip 
\begin{equation}
\rho _{crm}\left( x\right) =\inf \{t\in \mathbb{R}:x+tx_{0}\in P\}\;\text{\textit{for every}}\;x\in X\text{,} \tag{4}
\end{equation}
\textit{is a coherent risk measure with respect to the cone }$P$\ \textit{and the vector} $x_{0}$. \textit{Moreover, if} $P$ \textit{is closed then} $A_{\rho _{crm}}=P$.%
\footnote{In the sequel we will refer to the function $(4)$ as the (coherent) risk measure defined on $X$ with respect to the cone $P$ and the order unit $x_{0} $.}.
\end{lemma}

\noindent We state below (without proof), as a Lemma, Theorem 6 in Kountzakis et al. (2013). 
It will be invoked in the proof of our main result (i.e., Theorem 1).

\begin{lemma}
\textit{Let} $X$ \textit{be a Banach space ordered by the closed cone} $P$ \textit{and suppose that} $x_{0}\in P$ \textit{is an order unit of }$X$. 
\textit{If} $\rho $ \textit{is the risk measure defined on} $X$ \textit{with respect to the cone} $P\subseteq X$ \textit{and the order unit} $x_{0}\in P$, \textit{then for any} $x\in X$, $\rho \left( x\right) <0$ \textit{if and only if }$x$ \textit{is an interior point of }$P$.
\end{lemma}

\section{An axiomatization of preferences for which utility maximization is equivalent to risk minimization}

\noindent We will denote by $\ell _{\infty }^{C}$ the set of constant
sequences in $\ell _{\infty }$. Let $\succsim $ be a binary (preference)
relation on $\ell _{\infty }$. We state below, as axioms, several properties
that $\succsim $ is assumed to satisfy. Except for the risk-perception axiom
A.6., the other axioms we use are more or less imported from Gilboa and
Schmeidler (1989) to the framework under analysis. To be precise, our axiom
(A.4) part (a) is from Gilboa and Schmeidler (1989) but part (b) is not
contemplated therein. Also, it's easy to see that our axiom (A.3) implies
the continuity axiom posited in Gilboa and Schmeidler (1989). We will prove
that our axioms fully characterize a utility function $U$ such that $-U$ is
equal to a (robust) utility-based shortfall risk measure. In the remainder
of the paper keep in mind that $\succ $ and $\sim $ denote, respectively,
the asymmetric and symmetric parts of $\succsim $.\smallskip
\begin{axiom}[\textbf{A.1}]
\textit{Completeness and transitivity}: \textit{For all }$\mathbf{x},\mathbf{y}\in \ell _{\infty }$, $\mathbf{x}\succsim \mathbf{y}$ \textit{or} $\mathbf{y}\succsim \mathbf{x}$. 
\textit{Also, for all} $\mathbf{x}$, $\mathbf{y}$ \textit{and} $\mathbf{z}$ \textit{in} $\ell _{\infty }$, \textit{if} $\mathbf{x}\succsim \mathbf{y}$ \textit{and} $\mathbf{y}\succsim \mathbf{z}$ \textit{then} $\mathbf{x}\succsim \mathbf{z}$.
\end{axiom}

\begin{axiom}[\textbf{A.2}]
\textit{Certainty-Independence}: \textit{For all} $\mathbf{x},\mathbf{y}\in \ell _{\infty }$ \textit{and} $h\mathbf{1\in }\ell_{\infty }^{C}$%
\footnote{Where $h$ is any real number.}, 
\textit{and for all} $\alpha \in \left(0,1\right) $, $\mathbf{x\succsim y}$ \textit{if and only if} $\alpha \mathbf{x}+\left( 1-\alpha \right) h\mathbf{1\succsim }\alpha \mathbf{y} + \left(1-\alpha \right) h\mathbf{1}$.
\end{axiom}

\begin{axiom}[\textbf{A.3}]
\textit{Continuity}: \textit{For every} $\mathbf{x}\in \ell _{\infty }$, \textit{the subsets} $\left\{ \mathbf{y}\in \ell_{\infty }:\mathbf{y}\succsim \mathbf{x}\right\} $ \textit{and} $\left\{ \mathbf{z}\in \ell _{\infty }:\mathbf{x}\succsim \mathbf{z}\right\}$ \textit{are both norm-closed}.
\end{axiom}

\begin{axiom}[\textbf{A.4}]
\textit{Uniform monotonicity}: \textit{(a) For all} $\mathbf{x}$ \textit{and} $\mathbf{y}$ \textit{in} $\ell _{\infty}$, \textit{if }$x_{n}\geq y_{n}$ \textit{for every} $n\in \mathbb{N}$ \textit{then} $\mathbf{x\succsim y}$. 
\textit{(b) For all} $\mathbf{x}$\textit{\ and }$\mathbf{y}$\textit{\ in} $\ell _{\infty }$, \textit{if there exists some} $\epsilon >0$ \textit{such that }$x_{n}\geq y_{n}+\epsilon $ \textit{for every }$n\in \mathbb{N}$, \textit{then} $\mathbf{x}\succ \mathbf{y}$.
\end{axiom}

\begin{axiom}[\textbf{A.5}]
\textit{Uncertainty aversion}: \textit{For all} $\mathbf{x}$ \textit{and} $\mathbf{y}$ \textit{in} $\ell _{\infty }$ \textit{and} $\alpha $ $\in \left( 0,1\right) $, $\mathbf{x\sim y}$ \textit{implies} $\alpha \mathbf{x}+\left( 1-\alpha \right) \mathbf{y\succsim x}$.
\end{axiom}
Before stating the sixth and last axiom, a useful reminder is in
order: recall that%
\footnote{See, e.g., Theorem 14.9, the proof of Theorem 14.10, and Corollary 14.11 in Aliprantis and Border (2006).} 
to each finitely additive probability measure  $\pi $ on $\left( \mathbb{N},2^{\mathbb{N}}\right) $ there corresponds a unique positive linear functional $x_{\pi }^{\ast }$ in the topological dual of $\ell_{\infty}$ satisfying $x_{\pi }^{\ast }\left( \mathbf{x}\right) =\int_{\mathbb{N}}\mathbf{x}d\pi $ for all $\mathbf{x\mathbf{=}}\left(x_{n}\right) \mathbf{\in }\ell_{\infty}$. 
Thus, we can view $x_{\pi}^{\ast }\left(\mathbf{x}\right) $ as the expected value $\mathbb{E}_{\pi }[\mathbf{x}]$.

\begin{axiom}[\textbf{A.6}]\textit{Risk-perception}:
\begin{equation*}
\left\{\mathbf{x}\in \ell_{\infty }:\mathbf{x\succsim 0}\right\} =\left\{\mathbf{x}\in \ell_{\infty }:\mathbb{E}_{\pi }\left[ \mathbf{x}\right] \geq 0\text{ for all }\pi \in D\right\} =P,
\end{equation*}
\textit{where }$D$ \textit{is a closed and convex subset of} \textit{finitely additive probability measures on }$\left( \mathbb{N}, 2^{\mathbb{N}}\right) $.
\end{axiom}

\noindent Note that $D$ can be interpreted as the person-specific ambiguity set. 
This is because $\mathbb{E}_{\pi }\left[ \mathbf{x}\right] \geq 0$ for all $\pi \in D$ if and only if $\inf \left\{ \mathbb{E}_{\pi }\left[ \mathbf{x}\right] :\pi \in D\right\} \geq 0$.

\noindent Axioms (A.1) and (A.3) are rather restrictive but very standard properties in decision theory. 
For a discussion of axiom (A.2) we refer the reader to Gilboa and Schmeidler (1989). 
As in Gilboa and Schmeidler (1989) and Schmeidler (1989), Axiom (A.5) suggests, intuitively, that "smoothing" or averaging payoffs over random states of nature makes the decision-maker better off. 
Loosely speaking, as uncertainty decreases (in moving from $\mathbf{x}$ to $\alpha \mathbf{x}+\left(1-\alpha \right) \mathbf{y}$) the decision-maker is better-off. 
As to the risk-perception axiom, it will become clear when we prove Theorem 1 (see the appendix) that $P$ is the set of financial positions that have non-positive risk (hence, the set of acceptable positions) from the perspective of the risk measure corresponding to the utility function which represents the given preferences. 
This explains why we call $P$ the acceptance set. 
Axiom (A.6) captures the decision-maker's perception of risk.
In the  psychology literature on risk taking, it is well-understood that 
To see this, recall that the true probabilty measure on $\left( \mathbb{N}, 2^{\mathbb{N}}\right)$ is unknown.
Therefore, if $D$ is a singleton (i.e., no ambiguity) basically the decision-maker is a \enquote{confident} individual who is not afraid of making mistakes. 
In fact, if $D$ is a singleton the acceptance set $P$ is \enquote{large} and contains $\ell_{\infty}^{+}$, which signifies that the decision-maker has a low perception of risk: a \enquote{large} acceptance set means that \enquote{many} financial positions are not deemed risky. 
The polar opposite case is when $D$ is the set of all finitely additive probability measures on $\left(\mathbb{N}, 2^{\mathbb{N}}\right)$ (maximum ambiguity): the decision-maker is not \enquote{confident}, it wants to be on the safe side. 
In fact, if all conceivable probability measures are taken into account $P$ coincides with $\ell_{\infty}^{+}$, therefore $P$ is the smallest acceptance set. 
This signifies that the decision-maker has a high perception of risk. 
In fact, a \enquote{small} acceptance set means that \enquote{many} financial positions are perceived as risky.

\noindent Note that in our model ambiguity is positively related to risk-perception. 
Therefore, it is pertinent to refer to the significant body of research in psychology dealing with the concept of risk perception. 
To cite a recent one, we have been inspired by Weber (2004) who focuses on the influence of psychology in the decision making process under risk and uncertainty. 
He reviews the existing evidence in the literature which indicates that perceptions of risk are both \emph{subjective} and \emph{relative}.
Thus, since risk perception is subjective it differs across situations, persons, cultures and gender; since risk perception is also relative in nature, it depends on a standard of reference.
These features of risk perception provide richer set of explanations for the differences in risk taking observed in real life.



\noindent Basically, we add the risk-perception axiom to a system of axioms which is a slight modification of the one considered by Gilboa and Schmeidler (in a different setup, though). 
Therefore, it is logical to ask whether the risk-perception axiom is non-redundant. 
We claim that A.6 is actually non-redundant. 
To establish the claim it will suffice to prove that it is not the case that if axioms (A.1)-(A.5) hold true then axiom (A.6) is also true. 
Since the proof is easy, in the interest of brevity we will just sketch it out.

\begin{claim}
\textit{The risk-perception axiom is non-redundant.}
\end{claim}

\begin{proof}
Pick any finitely additive probability measure $\hat{\pi}$ which does not belong to $D$, and define $U:\ell _{\infty}\rightarrow \mathbb{R}$ by $U(\mathbf{x}):=\mathbb{E}_{\hat{\pi}}\left[\mathbf{x}\right]$. 
Next, define the binary relation $\succsim $ on $\ell_{\infty}$ as follows: for all $\mathbf{x},\mathbf{y}\in \ell_{\infty}$, $\mathbf{x}\succsim \mathbf{y}$ if and only if $U(\mathbf{x})\geq U(\mathbf{y})$. 
It is easy to show that $\succsim$ satisfies axioms A.1 through A.5. 
Therefore, to finish the proof we must show that A.6 fails to hold for the above-defined relation $\succsim$. 
To this end, it is enough to consider any vector $\mathbf{\hat{x}}\in \ell_{\infty}$ such that $\mathbf{\hat{x}}\notin P$ but $\mathbb{E}_{\hat{\pi}}\left[ \mathbf{\hat{x}}\right] \geq 0$. 
Such a vector exists because by assumption $\hat{\pi}\notin D$ and $D$ is a subset of the set of all finitely additive probability measures.
\end{proof}

\subsection{The main result\protect\medskip}

\noindent The proof of our main result is relegated to the appendix. Our
approach to the proof of the $(2)\Rightarrow (1)$ part of Theorem 1 is novel
and relies on some known facts about coherent risk measures (put together in
section $3$). The first half (at least) of the proof of the $(1)\Rightarrow
(2)$\emph{\ }part of Theorem 1 by and large mimics the proofs of Lemmas 3.2
and 3.3 in Gilboa and Schmeidler (1989). The final part is original and
relies, again, on some known properties of coherent risk measures in Banach
spaces.\smallskip

\begin{theorem}
\textit{Let} $\succsim$\textit{\ be a binary (preference) relation on }$\ell _{\infty }$\textit{. Let }$D$ \textit{be a
closed and convex subset of} \textit{finitely additive probability measures
on }$\left( \mathbb{N},2^{\mathbb{N}}\right) $. \textit{Then, the following
conditions are equivalent:\smallskip }

\noindent $(1)$ $\succsim$ \textit{satisfies axioms A.1-A.6.}

\noindent $(2)$\textit{There exists a function} $U:\ell_{\infty}\rightarrow \mathbb{R}$\textit{, given by }$U\left( \mathbf{x}\right) =\min \left\{ \mathbb{E}_{\pi }\left[ \mathbf{x}\right] :\pi \in D\right\} $, for all $\mathbf{x}\in \ell _{\infty }$, \textit{which is such that}
\begin{equation*}
\mathbf{x}\succsim \mathbf{y}\ \text{\textit{if and only if} }U(\mathbf{x})\geq U(\mathbf{y})\text{, \textit{for all}}\mathbf{x},\mathbf{y}\in \ell_{\infty}.
\end{equation*}%
\end{theorem}

\begin{proof}
See the appendix.
\end{proof}

\begin{remark}
\emph{Note that $-U$, where $U$ is the utility function in Theorem 1, is a special case of SR (i.e., a utility-based
shortfall risk measure). 
To see this, following Guo and Xu (2019)%
\footnote{See page 475 therein.} 
let $l:\mathbb{R}\rightarrow \mathbb{R}$ be a convex, increasing and non-constant function, and let $\lambda $ be a given constant in the interior of the range of $l$. 
Let $\mathbb{P}$ be a finitely additive probability measure on $\left(\mathbb{N}, 2^{\mathbb{N}}\right) $. 
Let $\mathbf{x}\in \ell _{\infty }$ be any financial position, i.e., a bounded random variable $\mathbf{x:}\left(\mathbb{N},2^{\mathbb{N}},\mathbb{P}\right) \rightarrow \mathbb{R}$. 
The SR of a financial position $\mathbf{x}\in \ell _{\infty }$ is defined as:}
\begin{equation*}
SR_{l,\lambda }^{\mathbb{P}}\left( \mathbf{x}\right) :=\inf \left\{ t\in \mathbb{R}:\mathbf{x+}t\mathbf{1\in }\mathcal{A}_{\mathbb{P}}\right\},
\end{equation*}%
\emph{where $\mathcal{A}_{\mathbb{P}}:=\left\{ \mathbf{x:}\left(\mathbb{N}, 2^{\mathbb{N}},\mathbb{P}\right) \rightarrow \mathbb{R}:\mathbf{x} \; \text{is bounded and }\mathbb{E}_{\mathbb{P}}\left[ l\left( -\mathbf{x}\left(n\right) \right) \right] \leq \lambda \right\} $. 
Now, set $l\left( x\right) =x$, $\lambda =0$, and $D=\left\{ \mathbb{P}\right\} $. 
Then, one can readily see that $\mathcal{A}_{\mathbb{P}}=P$.%
\footnote{Refer to the above axiom A.6.} 
Therefore, it follows from the formula displayed above and the proof of Theorem 1 that $SR_{l,\lambda }^{\mathbb{P}}\left( \mathbf{x}\right) =-U\left( \mathbf{x}\right) $. 
Similarly, it would be easy to verify that if $D$ is not a singleton $-U$ is equal to a distributionally robust SR.%
\footnote{See page 476 in Guo and Xu (2019) for the definition of distributionally robust SR.}}
\end{remark}

\begin{remark}
\emph{Within an expected-utility framework, where the
decision-maker's von Neumann-Morgenstern utility is a function of a control
variable and an exogenous random variable, Kimball (1990) develops a theory
of the optimal response of decision variables to risk, and, in particular, a
theory of precautionary saving. The author defines the concept of `prudence'
as a measure of the sensitivity of the optimal choice of a decision variable
to risk, and analyzes the determination of precautionary saving as an
example of the effect of risk on a decision variable. Since an increase in
precautionary saving, as a result of an increase in risk (i.e., a prudent
agent), may be interpreted as an increase in demand for `insurance', one may
say that the greater the demand for `insurance' in response to risk, the
more prudent the decision-maker. With this in mind, though our framework is
outside of the expected utility paradigm\footnote{%
So, it is not accurate to directly relate our results to Kimball's.}, we can
nonetheless draw an analogy between our risk-perception construct and
Kimball's notion of prudence. In fact, we show next that in our setting an
individual with maximum risk-perception is prudent, in the sense that such
decision-maker wants to be `fully insured', i.e., she demands a risk-free
(constant) financial position.}
\end{remark}

\begin{claim}
\textit{Consider a decision-maker with maximum risk-perception, i.e., $D$ is the set of all finitely additive probability measures on $\left( \mathbb{N},2^{\mathbb{N}}\right)$.
Suppose that the solution set of program $(1)$ is non-empty. 
Then, there exists a real number $\alpha \geq 0$ such that $\mathbf{c}=\alpha \mathbf{1}$ is a solution to program $(1)$.}
\end{claim}

\begin{proof}
Consider axiom A.6 and assume that $D$ is the set of all finitely additive probability measures on $\left( \mathbb{N}, 2^{\mathbb{N}}\right)$. 
Then, it should be clear that $P=\ell _{\infty }^{+}$.
Consequently, it follows from the proof of Theorem 1 (see the Appendix) that, for all $\mathbf{x}\in \ell _{\infty}$,
\begin{align*}
-U\left(\mathbf{x}\right) &=\inf \left\{ t\in \mathbb{R}:\mathbf{x}+t \mathbf{1}\in \ell _{\infty }^{+}\right\} =\inf \{t\in \mathbb{R}:x_{n}+t\geq 0\;\text{for every}\; n\} \\ 
&=\inf \{t\in \mathbb{R}:t\geq -x_{n}\;\text{for every}\;n\}\\ 
&=\sup_{n}\{(-x_{n})\}=-\inf_{n}\{(x_{n})\}.
\end{align*}%

\noindent This shows that the utility function is a Rawlsian function. Next,
suppose that there is a $\mathbf{c}^{\ast }\in \ell _{\infty }^{+}$ such
that $\inf \left\{ q\left( h\right) :h\in \mathbb{R}^{J}\text{ and }\mathbf{c%
}^{\ast }\leq \omega +R\left( h\right) \right\} \leq w_{0}$ and $%
\inf_{n}\{(c_{n}^{\ast })\}\geq \inf_{n}\left\{ (z_{n})\right\} $ for all
feasible $\mathbf{z}$. Obviously, $\inf_{n}\{(c_{n}^{\ast })\}$ is a
non-negative real number. So, let $\inf_{n}\{(c_{n}^{\ast })\}=\alpha \geq 0$
and define $\mathbf{c}:=\alpha \mathbf{1}$. Note that $\mathbf{c}\in \ell
_{\infty }^{+}$ and, by construction, $\mathbf{c}\leq \mathbf{c}^{\ast }$.
Therefore, also $\mathbf{c}$ is feasible for program $(1)$.\footnote{%
See the comments below $(2)$.} Furthermore, $U\left( \mathbf{c}\right)
=\alpha =U\left( \mathbf{c}^{\ast }\right) $. This establishes that the
constant sequence $\alpha \mathbf{1}$ is a solution to program $(1)$.
\end{proof}

\section{Concluding remarks\protect\medskip}

\noindent In this paper we have focused on a single decision-maker and
addressed the following question: how can a personality trait like
risk-perception be formalized, and how should preferences look like in order
for the traditional utility-maximization objective to be aligned with the
risk-minimization objective, where risk is measured by a known risk measure
that allows for various degrees of risk-perception? Our Theorem 1 tackles
this question by developing an axiomatic foundation of preferences for which
utility-maximization is equivalent to minimizing a utility-based shortfall
risk measure. Our axiomatization hinges on a novel axiom in decision theory,
namely the risk-perception axiom, and adapts (and expands) the system of
axioms in Gilboa and Schmeidler (1989) to preference relations defined on
financial positions or contingent payoffs, where the space of financial
positions is the set of bounded real-valued sequences. Our model and main
result, Theorem 1, may be framed quite naturally also in the space of
financial positions given by the set of continuous real-valued functions on
the closed unit interval.

\section*{Appendix}

\begin{proof} 
[Proof of Theorem 1]

\noindent $(2)\Rightarrow (1)$:

\noindent Let $U\left( \mathbf{x}\right) =\min \left\{ \mathbb{E}_{\pi }%
\left[ \mathbf{x}\right] :\pi \in D\right\} $, for all $\mathbf{x}\in \ell
_{\infty }$, and define the binary (preference) relation $\succsim $ on%
\textit{\ }$\ell _{\infty }$ by $\mathbf{x}\succsim \mathbf{y}\ $if and only
if $U(\mathbf{x})\geq U(\mathbf{y})$, for all $\mathbf{x},\mathbf{y}\in \ell
_{\infty }$. Since $U$ is real-valued and $D$ is closed, we have
that\smallskip

\begin{equation}
\min \left\{ \mathbb{E}_{\pi }\left[ \mathbf{x}\right] :\pi \in D\right\}
=\inf \left\{ \mathbb{E}_{\pi }\left[ \mathbf{x}\right] :\pi \in D\right\} ,%
\text{ for all }\mathbf{x}\in \ell _{\infty \text{.}}\text{.}  \tag{5}
\end{equation}%
\smallskip

\noindent Using $(5)$, some straightforward algebra shows the
following:\smallskip

\begin{equation*}
-U\left( \mathbf{x}\right) =-\inf \left\{ \mathbb{E}_{\pi }\left[ \mathbf{x}%
\right] :\pi \in D\right\} =\inf \left\{ t\in \mathbb{R}:\mathbf{x}+t\mathbf{%
1}\in P\right\} \text{,}
\end{equation*}%
\smallskip

\noindent where $P=\left\{ \mathbf{x}\in \ell _{\infty }:\mathbb{E}_{\pi }%
\left[ \mathbf{x}\right] \geq 0\text{ for all }\pi \in D\right\} .$ Since $P$
is a closed cone containing $\ell _{\infty }^{+}$, and $\mathbf{1\in }P$ is
an order unit\footnote{%
With respect to both $\ell _{\infty }^{+}$ and $P$.}, by Lemma 1 and the
above equality we have that $-U:\ell _{\infty }\rightarrow \mathbb{R}$ is a
coherent risk measure, with respect to $\ell _{\infty }^{+}$\ and the vector 
$\mathbf{1}$, and $\mathcal{A}_{-U}=P$. Therefore, it follows from the four
defining properties of coherent risk measures (in section $3$) and the
definition of acceptance set that\smallskip

\begin{equation*}
\begin{tabular}{l}
$\mathbf{x}-\mathbf{y\in }\ell _{\infty }^{+}$ implies $U\left( \mathbf{x}%
\right) \geq U\left( \mathbf{y}\right) $. \ $(i)$\smallskip \\ 
$U\left( \mathbf{x+}t\mathbf{1}\right) =U\left( \mathbf{x}\right) +t$, for
all $\mathbf{x}\in \ell _{\infty }$ and $t\in \mathbb{R}$. \ $(ii)$\smallskip
\\ 
$U\left( \mathbf{x+y}\right) \geq U\left( \mathbf{x}\right) +U\left( \mathbf{%
y}\right) $, for all $\mathbf{x},\mathbf{y}\in \ell _{\infty }$. \ $(iii)$%
\smallskip\  \\ 
$U\left( \lambda \mathbf{x}\right) =\lambda U\left( \mathbf{x}\right) $, for
all $\mathbf{x}\in \ell _{\infty }$ and $\lambda \geq 0$. \ $(iv)$\smallskip
\\ 
$\left\{ \mathbf{x}\in \ell _{\infty }:U\left( \mathbf{x}\right) \geq
0\right\} =A_{-U}=P$ \ $(v)$.%
\end{tabular}%
\end{equation*}%
Of course, $\succsim $ satisfies axiom A.1. Using the above
properties $(ii)$ and $(iv)$ it is very easy to see that $\succsim $
satisfies axiom A.2. By the above properties $(iii)$ and $(iv)$, $-U$ is a
sublinear function and one can readily see that it is bounded on every
neighborhood of zero. Hence, $U$ is continuous in the supremum norm-topology%
\footnote{%
See, e.g., Lemma 5.51 in Aliprantis and Border (2006).}. In turn, continuity
of $U$ easily implies that $\succsim $ satisfies axiom A.3. That $\succsim $
satisfies axiom A.4. part $(a)$ is an immediate consequence of property $(i)$%
. As to part $(b)$ of axiom A.4., suppose that there is some $\epsilon >0$
such that $x_{n}-y_{n}\geq \epsilon $ for every $n\in \mathbb{N}$. Then, $%
\mathbf{x}-\mathbf{y}$ lies in the norm-interior of $\ell _{\infty }^{+}$.
Since $P\supseteq \ell _{\infty }^{+}$, it also belongs to the norm-interior
of $P$. Therefore, since $-U\left( \mathbf{x}\right) =\inf \left\{ t\in 
\mathbb{R}:\mathbf{x}+t\mathbf{1}\in P\right\} $ (see above) it follows from
Lemma 2 that $U\left( \mathbf{x}-\mathbf{y}\right) >0$. On the other hand,
as an easy consequence of property $(iii)$ we get $U\left( \mathbf{x}\right)
-U\left( \mathbf{y}\right) \geq U\left( \mathbf{x}-\mathbf{y}\right) $,
which implies $U\left( \mathbf{x}\right) >U\left( \mathbf{y}\right) $, hence 
$\mathbf{x}\succ \mathbf{y}$. This establishes that $\succsim $ satisfies
part $(b)$ of axiom A.4. Next, a very straightforward application of
properties $(iii)$ and $(iv)$ shows that $\succsim $ satisfies axiom A.5.
Finally, property $(iv)$ clearly yields $U\left( \mathbf{0}\right) =0$.
Thus, $\succsim $ satisfies axiom A.6 as an obvious implication of property $%
(v)$.\smallskip

\noindent $(1)\Rightarrow (2)$: To improve readability, the proof will be
divided into three steps.

\noindent \textit{Step 1 (Existence of a utility function representing
preferences):}

\noindent Let $\succsim $\ be a preference relation on $\ell _{\infty }$
satisfying axioms A.1-A.6. Pick any arbitrary $\mathbf{x\in }\ell _{\infty }$%
, and define the following sets:
\begin{eqnarray*}
A &=&\left\{ t\in \mathbb{R}:t\mathbf{1}\succsim \mathbf{x}\right\} \text{%
\smallskip } \\
B &=&\left\{ t\in \mathbb{R}:\mathbf{x}\succsim t\mathbf{1}\right\} \text{.}
\end{eqnarray*}%
We claim that $A$ and $B$ are both non-empty. To see this, note
that, by definition of $\ell _{\infty }$, there exists some $M>0$ such that $%
-M\leq x_{n}\leq M$ for every $n\in \mathbb{N}$. Thus, by A.4 part (a) we
get $M\mathbf{1\succsim x}\succsim -M\mathbf{1}$. Therefore, $M\in A$, and $%
-M\in B$. Moreover, using axiom A.3 it's easy to check that $A$ and $B$ are
both closed. Furthermore, completeness of preferences (see axiom A.1)
implies $\mathbb{R=}A\cup B$. Thus, since $\mathbb{R}$ is connected, we must
have that $A\cap B\neq \varnothing $. This readily implies (see the above
definitions of sets $A$ and $B$) that there exists a real number $t_{\mathbf{%
x}}$ such that $t_{\mathbf{x}}\mathbf{1\sim x}$. Such a number is unique.
For, suppose, by way of obtaining a contradiction, that there is some $s_{%
\mathbf{x}}\neq t_{\mathbf{x}}$ such that $s_{\mathbf{x}}\mathbf{1\sim x}$,
and assume, without any loss of generality, that $s_{\mathbf{x}}>t_{\mathbf{x%
}}$. Then, on the one hand A.1 implies $s_{\mathbf{x}}\mathbf{1\sim }t_{%
\mathbf{x}}\mathbf{1}$, and on the other hand it follows from A.4 part (b)
that $s_{\mathbf{x}}\mathbf{1}\succ t_{\mathbf{x}}\mathbf{1}$, which yields
the desired contradiction. So, we have established the existence of a
function $U:\ell _{\infty }\rightarrow \mathbb{R}$ defined by
\begin{equation*}
\ell _{\infty }\ni \mathbf{x\mapsto }U\left( \mathbf{x}\right) :=t_{\mathbf{x%
}}\text{ with }t_{\mathbf{x}}\mathbf{1\sim x}\text{.}
\end{equation*}
Moreover, using axioms A.1 and A.4 it is a simple exercise to
verify that $\mathbf{x}\succsim \mathbf{y}$ if and only if $U\left( \mathbf{x%
}\right) \geq U\left( \mathbf{y}\right) $. This proves that the above $%
U:\ell _{\infty }\rightarrow \mathbb{R}$ is a utility function representing
the given preferences. Incidentally, observe that, by construction of the
utility function, $U\left( c\mathbf{1}\right) =c$ for every real number $c$.%
\footnote{This fact will be used again in the rest of the proof without further notice.} 

\textit{Step 2 (Showing that the utility function is a functional a la Gilboa and Schmeidler):}

\noindent We claim that axioms (A.1)-(A.5) imply the existence of a
non-empty, closed and convex set $C$ of finitely additive probability
measures on $2^{\mathbb{N}}$ such that\smallskip \emph{\ 
\begin{equation*}
U\left( \mathbf{x}\right) =\min \left\{ \mathbb{E}_{\pi }\left[ \mathbf{x}%
\right] :\pi \in C\right\} ,\text{ for all }\mathbf{x}\in \ell _{\infty }.
\end{equation*}}
To prove the claim, we begin by noting that A.4 part $(a)$
immediately implies that the utility function is monotonic, i.e., satisfies
property $(i)$ listed above. Next, we will prove that the utilty function $U$
satisfies property $(iv)$ listed above. That $(iv)$ holds true for $\lambda
=0$ and $\lambda =1$ is obvious. So, assume, to begin with, that $0<\lambda
<1$, and pick any $\mathbf{x\in }\ell _{\infty }$. Note that $U\left(
\lambda \mathbf{x}\right) =t_{\lambda \mathbf{x}}$ where $t_{\lambda \mathbf{%
x}}\mathbf{1\sim }\lambda \mathbf{x}$, and $U\left( \mathbf{x}\right) =t_{%
\mathbf{x}}$ where $t_{\mathbf{x}}\mathbf{1\sim x}$. Since $t_{\mathbf{x}}%
\mathbf{1\sim x}$ and $0<\lambda <1$, it follows from A.2 that $\lambda t_{%
\mathbf{x}}\mathbf{1\sim }\lambda \mathbf{x}$. On the other hand, we know
that $t_{\lambda \mathbf{x}}\mathbf{1\sim }\lambda \mathbf{x}$, therefore
A.1 implies $t_{\lambda \mathbf{x}}\mathbf{1\sim }\lambda t_{\mathbf{x}}%
\mathbf{1}$, hence $U\left( t_{\lambda \mathbf{x}}\mathbf{1}\right)
=t_{\lambda \mathbf{x}}=U\left( \lambda t_{\mathbf{x}}\mathbf{1}\right)
=\lambda t_{\mathbf{x}}$, which finally implies $U\left( \lambda \mathbf{x}%
\right) =\lambda U\left( \mathbf{x}\right) $. We are left with showing that
property $(iv)$ holds true also for $\lambda >1$, as follows: we know that $%
t_{\lambda \mathbf{x}}\mathbf{1\sim }\lambda \mathbf{x}$, $t_{\mathbf{x}}%
\mathbf{1\sim x}$, and that $0<\frac{1}{\lambda }<1$. Therefore, by A.2 we
get $\frac{1}{\lambda }t_{\lambda \mathbf{x}}\mathbf{1\sim x}$, so A.1
readily implies $\frac{1}{\lambda }t_{\lambda \mathbf{x}}\mathbf{1\sim }t_{%
\mathbf{x}}\mathbf{1}$. Thus,
\begin{equation*}
t_{\mathbf{x}}=U\left( t_{\mathbf{x}}\mathbf{1}\right) =U\left( \frac{1}{%
\lambda }t_{\lambda \mathbf{x}}\mathbf{1}\right) =\frac{1}{\lambda }U\left(
t_{\lambda \mathbf{x}}\mathbf{1}\right) =\frac{1}{\lambda }t_{\lambda 
\mathbf{x}}, 
\end{equation*}
which yields immediately $\lambda U\left( \mathbf{x}\right)
=U\left( \lambda \mathbf{x}\right) $. Next we show that the utility function 
$U$ enjoys property $(ii)$ listed above. Toward this end, pick any $\mathbf{x%
}\in \ell _{\infty }$ and $t\in \mathbb{R}$. 
Define
\begin{equation}
\beta =U\left( 2\mathbf{x}\right) =2U\left( \mathbf{x}\right) \text{,} 
\tag{6}
\end{equation}%
and note that $\beta =t_{2\mathbf{x}}$ with $\beta \mathbf{1\sim }2%
\mathbf{x}$. Observe that $2t\mathbf{1\in }\ell _{\infty }^{C}$. Therefore,
letting $\alpha =\frac{1}{2}$ and using A.2, we get $\frac{1}{2}2\mathbf{x+}%
\frac{1}{2}2t\mathbf{1\sim }\frac{1}{2}\beta \mathbf{1+}\frac{1}{2}2t\mathbf{%
1}$, i.e., $\mathbf{x+}t\mathbf{1\sim }(\frac{1}{2}\beta +t)\mathbf{1}$. The
previous condition and $(6)$ readily imply
\begin{equation*}
U\left( \mathbf{x+}t\mathbf{1}\right) =U\left( (\frac{1}{2}\beta +t)\mathbf{1%
}\right) =\frac{1}{2}\beta +t=U\left( \mathbf{x}\right) +t\text{,}
\end{equation*}
as was to be proven. We proceed to show that the utility function
is superadditive (property $(iii)$ listed above). To this end, a
straightforward application of property $(iv)$ reveals that it will suffice
to prove that for all $\mathbf{x}$ and $\mathbf{y}$ in $\ell _{\infty }$,
\begin{equation*}
U\left( \frac{1}{2}\mathbf{x+}\frac{1}{2}\mathbf{y}\right) \geq \frac{1}{2}%
U\left( \mathbf{x}\right) +\frac{1}{2}U\left( \mathbf{y}\right) \text{.}
\end{equation*}
To establish the above inequality, we pick any $\mathbf{x}$ and $%
\mathbf{y}$ in $\ell _{\infty }$ and distinguish three exhaustive cases.
Case $(a)$: $U\left( \mathbf{x}\right) =U\left( \mathbf{y}\right) $; case $%
(b)$: $U\left( \mathbf{x}\right) >U\left( \mathbf{y}\right) $; case $(c)$: $%
U\left( \mathbf{x}\right) <U\left( \mathbf{y}\right) $. We first deal with
case $(a)$, as follows: $U\left( \mathbf{x}\right) =U\left( \mathbf{y}%
\right) $ implies $\mathbf{x\sim y}$. Therefore, it follows from axiom A.5
that $\frac{1}{2}\mathbf{x+}\frac{1}{2}\mathbf{y\succsim x}$. Hence, $%
U\left( \frac{1}{2}\mathbf{x+}\frac{1}{2}\mathbf{y}\right) \geq U\left( 
\mathbf{x}\right) =\frac{1}{2}U\left( \mathbf{x}\right) +\frac{1}{2}U\left( 
\mathbf{y}\right) $, as was to be proven. Regarding case $(b)$, let
\begin{equation}
t=U\left( \mathbf{x}\right) -U\left( \mathbf{y}\right) \text{,}  \tag{7}
\end{equation}
and set $\mathbf{c}=\mathbf{y+}t\mathbf{1}$. By property $(ii)$
and $(7)$ above we get
\begin{equation*}
U\left( \mathbf{c}\right) =U\left( \mathbf{y+}t\mathbf{1}\right) =U\left( 
\mathbf{y}\right) +t=U\left( \mathbf{y}\right) +U\left( \mathbf{x}\right)
-U\left( \mathbf{y}\right) =U\left( \mathbf{x}\right) \text{.}
\end{equation*}
Now, using again property $(ii)$ we see that
\begin{equation}
U\left( \frac{1}{2}\mathbf{x+}\frac{1}{2}\mathbf{c}\right) =U\left( \frac{1}{%
2}\mathbf{x+}\frac{1}{2}\mathbf{y+}\frac{1}{2}t\mathbf{1}\right) =U\left( 
\frac{1}{2}\mathbf{x+}\frac{1}{2}\mathbf{y}\right) +\frac{1}{2}t\text{.} 
\tag{8}
\end{equation}%
\smallskip

\noindent On the other hand, since $U\left( \mathbf{c}\right) =U\left( 
\mathbf{x}\right) $ we can rely on the previous case $(a)$ and invoke
property $(ii)$ together with $(8)$ above to conclude:\smallskip

\begin{equation*}
\begin{tabular}{l}
$U\left( \frac{1}{2}\mathbf{x+}\frac{1}{2}\mathbf{y}\right) +\frac{1}{2}%
t=U\left( \frac{1}{2}\mathbf{x+}\frac{1}{2}\mathbf{c}\right) \geq \frac{1}{2}%
U\left( \mathbf{x}\right) +\frac{1}{2}U\left( \mathbf{c}\right) $\smallskip
\\ 
$=\frac{1}{2}U\left( \mathbf{x}\right) +\frac{1}{2}U\left( \mathbf{y+}t%
\mathbf{1}\right) =\frac{1}{2}U\left( \mathbf{x}\right) +\frac{1}{2}\left[
U\left( \mathbf{y}\right) +t\right] =\frac{1}{2}U\left( \mathbf{x}\right) +%
\frac{1}{2}U\left( \mathbf{y}\right) +\frac{1}{2}t$,%
\end{tabular}%
\end{equation*}

\noindent hence $U\left( \frac{1}{2}\mathbf{x+}\frac{1}{2}\mathbf{y}\right)
\geq \frac{1}{2}U\left( \mathbf{x}\right) +\frac{1}{2}U\left( \mathbf{y}%
\right) $, as was to be proven. As for case $(c)$, it turns out it can be
handled very similarly to case $(b)$.\smallskip

\noindent So far we have proven that axioms (A.1)-(A.5) give rise to a
utility function $U$, representing preferences, which satisfies properties $%
(i)-(iv)$ and which is such that $U\left( \mathbf{1}\right) =1$.\footnote{%
As such, $-U$ is a coherent risk measure (satisfying $U\left( \mathbf{1}%
\right) =1$) defined on $\ell _{\infty }$ ordered by $\ell _{\infty }^{+}$
and with $\mathbf{1}$ as an order unit.}\emph{\ }Therefore, by Lemma 3.5 in
Gilboa and Schmeidler (1989) there\emph{\ }exists a non-empty, closed and
convex set $C$ of finitely additive probability measures on $2^{\mathbb{N}}$%
such that\smallskip \emph{\ 
\begin{equation}
U\left( \mathbf{x}\right) =\min \left\{ \mathbb{E}_{\pi }\left[ \mathbf{x}%
\right] :\pi \in C\right\} ,\text{ for all }\mathbf{x}\in \ell _{\infty }. 
\tag{9}
\end{equation}%
\smallskip }

\noindent \textit{Step 3 (Pinpointing the ambiguity set }$C$\textit{%
):\smallskip }

\noindent In what follows we show that the risk-perception axiom A.6 pins
down the ambiguity set (i.e., $C=D$), hence the exact form of the utility
function $U$.\smallskip

\noindent We know from Step 2 that $-U$ is a coherent risk measure on $\ell
_{\infty }$, when the latter is ordered by $\ell _{\infty }^{+}$ and $%
\mathbf{1}$ is an order unit. Therefore, by $(3)$ above we get $-U\left( 
\mathbf{x}\right) =\inf \left\{ t\in \mathbb{R}:\mathbf{x}+t\mathbf{1}\in 
\mathcal{A}_{-U}\right\} $ for all $\mathbf{x\in }\ell _{\infty }$. Also, it
readily follows from the definition of acceptance set and axiom A.6 that $%
\mathcal{A}_{-U}=P=\left\{ \mathbf{x}\in \ell _{\infty }:\mathbb{E}_{\pi }%
\left[ \mathbf{x}\right] \geq 0\text{ for all }\pi \in D\right\} $.
Therefore,\smallskip

\begin{equation}
-U\left( \mathbf{x}\right) =\inf \left\{ t\in \mathbb{R}:\mathbf{x}+t\mathbf{%
1}\in P\right\} \text{.}  \tag{10}
\end{equation}%
\smallskip

\noindent Performing some straightforward algebra in $(10)$ we get $U\left( 
\mathbf{x}\right) =\inf \left\{ \mathbb{E}_{\pi }\left[ \mathbf{x}\right]
:\pi \in D\right\} ,$ for all $\mathbf{x}\in \ell _{\infty }$. Since every $%
\pi \in D$ is a probability measure and any $\mathbf{x}\in \ell _{\infty }$
is bounded, the $\inf $ in the previous equation is a real number.
Furthermore, because $D$ is closed, by assumption, such infimum is actually
attained. Therefore, we come to:\smallskip

\begin{equation}
U\left( \mathbf{x}\right) =\min \left\{ \mathbb{E}_{\pi }\left[ \mathbf{x}%
\right] :\pi \in D\right\} ,\text{for all }\mathbf{x}\in \ell _{\infty }. 
\tag{11}
\end{equation}%
\smallskip

\noindent Finally, one can readily see that preferences satisfying axiom
(A.6) are non-degenerate\footnote{%
For the concept of non-degeneracy and its implications for uniqueness of the
ambiguity set $C$, see pages 144-145 in Gilboa and Schmeidler (1989).}.
Thus, as in Gilboa and Schmeidler (1989) the ambiguity set $C$ in $(9)$
above must be unique. Therefore, since $D$ is by assumption closed and
convex, $(9)$ and $(11)$ imply $C=D$. The proof is now complete.
\end{proof}

\noindent \textbf{DECLARATIONS}\medskip

\noindent \textbf{Funding and/or Conflicts of interests/Competing interests} 

The authors have no relevant financial or non-financial interests to disclose.

The authors declare that they have no conflict of interest.

The authors have no competing interests to declare that are relevant to the content of this article.

\noindent \textbf{Data availability }We do not analyze or generate any data
sets because our work proceeds within a theoretical and mathematical
approach.\medskip

\noindent \textbf{References}\medskip

\noindent Aliprantis, C. D., \& Border, K. C. (2006). \textit{Infinite Dimensional Analysis A Hitchhiker's Guide}. Springer, New York, Third edition.\smallskip

\noindent Clark, S. A. (1993). The valuation problem in arbitrage price theory, \textit{Journal of Mathematical Economics} 22, 463-478.\smallskip

\noindent Gilboa I., \& Schmeidler, D. (1989). Maximin expected utility with non-unique prior, \textit{Journal of Mathematical Economics}, 18, 141-153.\smallskip

\noindent Gundel, A., \& Weber, S. (2007). Robust utility maximization with limited downside risk in incomplete markets, \textit{Stochastic Processes and their Applications,} 117, 1663--1688.\smallskip

\noindent Guo, S., \& Xu, H. (2019). Distributionally robust shortfall risk optimization model and its approximation, \textit{Mathematical Programming}, 174, 473-498.\smallskip

\noindent Kimball, M. S. (1990). Precautionary Saving in the Small and in the Large,\textit{\ Econometrica} 58, 53-73.\smallskip

\noindent Kountzakis, C., \& Polyrakis, I. A. (2013). Coherent risk measures in general economic models and price bubbles, \textit{Journal of Mathematical Economics}, 49 (3), 201-209.\smallskip

\noindent LeRoy, S. F., \& Werner, J. (2014). \textit{Principles of Financial Economics}, second edition, Cambridge University Press.\smallskip

\noindent Ross, S. (1978). A Simple Approach to the Valuation of Risky Streams, \textit{The Journal of Business}\ 51, 453--475.\smallskip

\noindent Schmeidler, D. (1989). Subjective Probability and Expected Utility without Additivity, \textit{Econometrica}, Vol. 57, No. 3, 571-587.

\noindent Weber, E. U. (2004). Perception Matters: Psychophysics for Economists. In J. Carrillo and I. Brocas (Eds.), Psychology and Economics (165-176). Oxford, UK: Oxford University Press.

\end{document}